\documentclass[aps,twocolumn,showpacs,pra,citeautoscript,amsmath,amssymb,floatfix,superscriptaddress, bbm, changes,10pt]{revtex4-1}
\pdfoutput=1
\usepackage{amsmath,bbold}
\usepackage{amssymb}
\usepackage{epsfig}
\usepackage{color}
\usepackage{amsthm}
\usepackage{hyperref}
\usepackage{appendix}
\usepackage{cancel}
\usepackage{stmaryrd}
\usepackage{soul}

\def\<{\langle}
\def\>{\rangle}

\def\compo{}

\newcommand{\be}{\begin{eqnarray} \begin{aligned}}
\newcommand{\ee}{\end{aligned} \end{eqnarray} }
\newcommand{\benn}{\begin{eqnarray*} \begin{aligned}}
\newcommand{\eenn}{\end{aligned} \end{eqnarray*} }

\newcommand{\ben}{\begin{eqnarray} \begin{aligned}}
\newcommand{\een}{\end{aligned} \end{eqnarray} }

\newcommand{\bc}{\begin{center}}
\newcommand{\ec}{\end{center}}

\newcommand{\id}{\mathbb{1}}

\newcommand{\tr}{\mathop{\mathsf{tr}}\nolimits}

\newcommand{\e}{\mathrm{e}}

\newcommand{\beq}{\begin{eqnarray} \begin{aligned}}
\newcommand{\eeq}{\end{aligned} \end{eqnarray} }
\newcommand{\bea}{\begin{array}}
\newcommand{\eea}{\end{array}}

\newcommand{\bee}{\begin{enumerate}}
\newcommand{\eee}{\end{enumerate}}
\newcommand{\bei}{\begin{itemize}}
\newcommand{\eei}{\end{itemize}}
\newtheorem{theorem}{Theorem}

\newtheorem{lemma}[theorem]{Lemma}

\usepackage{amsfonts}

\def\id{\mathbb{1}}

\def\01{\{0,1\}}

\newcommand{\ket}[1]{|#1\rangle}
\newcommand{\bra}[1]{\langle#1|}

\newcommand{\proj}[1]{|#1\rangle\langle#1|}
\newcommand{\ketbra}[2]{|#1\rangle\langle#2|}

\def\<{\langle}
\def\>{\rangle}

\newtheorem*{rep@theorem}{\rep@title}
\newcommand{\newreptheorem}[2]{%
\newenvironment{rep#1}[1]{%
 \def\rep@title{#2 \ref{##1} (restatement)}%
 \begin{rep@theorem}}%
 {\end{rep@theorem}}}
\makeatother

\newtheorem{result}{Result}
\def\e{\mathrm{e}}
\def\i{\mathrm{i}}
\def\E{\mathcal{E}}

\newreptheorem{thm}{Theorem}
\newreptheorem{lem}{Lemma}

\begin{document}

\title{
Fluctuating work: From quantum thermodynamical identities to a second law equality}

\author{\'{A}lvaro M. Alhambra}
%\email{alvaro.alhambra.14@ucl.ac.uk}
\affiliation{Department of Physics and Astronomy, University College London, Gower Street, London WC1E 6BT, United Kingdom}
\author{Lluis Masanes}
\affiliation{Department of Physics and Astronomy, University College London, Gower Street, London WC1E 6BT, United Kingdom}
\author{Jonathan Oppenheim}
%\email{j.oppenheim@ucl.ac.uk}
\affiliation{Department of Physics and Astronomy, University College London, Gower Street, London WC1E 6BT, United Kingdom}
%\affiliation{Department of Computer Science and Centre for Quantum Technologies, National University of Singapore, Singapore 119615}
\author{Christopher Perry}
%\email{christopher.perry.12@ucl.ac.uk}
\affiliation{Department of Physics and Astronomy, University College London, Gower Street, London WC1E 6BT, United Kingdom}
\affiliation{Department of Mathematical Sciences, University of Copenhagen, Universitetsparken 5, 2100 Copenhagen, Denmark}

\begin{abstract}
We investigate the connection between recent results in quantum thermodynamics and fluctuation relations by adopting a fully quantum mechanical description of thermodynamics. By
including a work system whose energy is allowed to fluctuate, we derive a set of equalities which all thermodynamical transitions have to satisfy.
This extends the condition for maps to be Gibbs-preserving to the case of fluctuating work, providing a more general characterisation of maps
commonly used in the information theoretic approach to thermodynamics.
For
final states, block diagonal in the energy basis, this set of equalities 
%
%extend the notion of operations which preserve the thermal state, to the fluctating work
%setting, deriving a set of equalities which characterise them. These form   
are necessary and sufficient conditions for a thermodynamical state transition to be possible.
The conditions serves as a parent equation which
can be used to derive a number of results. These include writing the second law of thermodynamics as an equality featuring a fine-grained notion of the free energy.
It also yields a generalisation of the Jarzynski fluctuation theorem which holds for arbitrary initial states, and under the most general manipulations allowed
by the laws of quantum mechanics.
Furthermore, we show that each of these relations can be seen as
the quasi-classical limit of three fully quantum identities. This allows us to consider the free energy as an operator, and 
allows one to obtain more general and fully quantum fluctuation relations from the information theoretic approach to quantum thermodynamics.
%
%for a thermodynamical transition to happen in the
\end{abstract}
\maketitle

\section{Introduction}
The second law of thermodynamics governs what state transformations are possible regardless of the details of the interactions. As such, it is arguably the law of physics with the broadest applicability, relevant for
situations as varied as subatomic collisions, star formation, biological processes, steam engines, molecular motors and cosmology.  For a system which could be placed in contact with a thermal reservoir at temperature $T$, the second law can be expressed as
an inequality of the form
\begin{equation}
  \langle w \rangle  \le F(\rho)-F(\rho')
 \end{equation}
where the free energy is $F(\rho)= \tr H\rho-T S(\rho)$, the entropy is $S(\rho)=-\sum_sP(s)\log P(s)$, with $P(s)$ the probability that the system
has energy level $\ket{s}$, $H$ is the Hamiltonian of the system, and $ \langle w \rangle$ is the average work done by the system when it transitions from $\rho$ to $\rho'$. The free energy
is a scalar, and can be thought of as an average quantity. Here, we will see that by thinking of the free energy as an operator or random variable, one can derive a much stronger classical version of the second law which is an equality
\begin{align}
  \left\langle \e^{\beta \left(f_{s'} -f_s +w \right)} \right\rangle = 1
%\label{eq:secondlaweq}
\end{align}
in terms of a fine-grained free energy
\begin{align}
f_s=E_s+T \log P(s)
\end{align}
that can be considered as a random variable occurring with probability $P(s)$ and whose average value is the ordinary scalar free energy $F= \langle f_s \rangle$. Here, initial energy levels are given by  $E_s=\tr \proj{s} H$, 
%with $P(s)=\tr \proj{s}\rho$
while the energy levels $E_{s'}$ correspond to the final Hamiltonian $H'$.
Although the term $-T \log P(s)$ is not defined for $P(s)=0$, all its moments are. We will see that
this equality version of the second law can be thought of as a simple consequence of a much stronger family of equalities and quantum identities. We will also see that if we Taylor expand the exponential in the above equality, we obtain not only the standard inequality version of the second law, but in addition, an infinite set of higher order inequalities. These can be thought of as corrections to the standard inequality.

This second law equality is valid for transitions between any two states as long as the initial state is diagonal in the energy eigenbasis, and when work is considered as the change in energy of some work system or {\it weight}. As such, although it is a greatly strengthened form of the second law, it is of a more classical nature, reminiscent of the fluctuation theorems of Jarzynski~\cite{jarzynski1997nonequilibrium} and Crooks~\cite{crooks1999entropy} where it is required that a measurement is performed on the initial and final state.  We are however, not only able to get a more general classical version of the Jarzynski fluctuation theorem, valid for any initial state, but we are also able to derive two fully quantum identities which reduce in the classical limit to these classical generalisations of the Jarzynski equation and the second law.

What's more, even our classical fluctuation theorems are derived from a fully quantum mechanical point of view, and are thus valid for any quantum process. Previous derivations, assumed a particular form of Markovian classical trajectories
(e.g. assuming a Langevin equation or classical trajectory in the context of classical stochastic thermodynamics) \cite{crooks1999entropy,seifert2012stochastic,seifert2005entropy,chernyak2006path,hatano2001steady,sagawa2010generalized,schumacher-eth}.
 A quantum mechanical derivation of the standard
Jarzynski equation has been done in numerous works \cite{tasaki2000jarzynski,campisi2009fluctuation,talkner2009fluctuation,campisi2011colloquium,albash2013fluctuation,rastegin2014jarzynski},
but in these one is usually limited to initial thermal states, and also must resort to energy measurements on the system in order to define work.
A derivation of fluctuation theorems for classical trajectories between arbitrary initial and final states has been performed in the case of erasure and a degenerate Hamiltonian \cite{schumacher-eth}. A derivation with
arbitrary initial and final states was also undertaken in \cite{manzano2015nonequilibrium}, for a family of maps which 
go slightly beyond the classical case.  

Here, we adopt a fully general and quantum mechanical treatment and derive a fully quantum identity, which reduces to the generalised fluctuation relation
\begin{equation}
  \left\langle 
  \e^{\beta( w -f_s)} \right\rangle
  = Z'_S
%\label{eq:genjar}
\end{equation}
when the initial state is diagonal in the energy eigenbasis. We call such states, i.e. those that satisfy $[\rho_S,H_S]=0$, {\it quasi-classical states}, 
and the above fluctuation relation 
%Equation \eqref{eq:genjar},
is valid for arbitrary initial and final states of this form and for any quantum thermodynamical process.
When the initial state is thermal, we further have $\e^{-\beta f_s}  = Z_S$ for all $s$ which gives Jarzynski's equation in its usual form
\begin{equation}
  \left\langle 
  \e^{\beta w} \right\rangle
  = \frac {Z'_S} {Z_S}
  \ .
\end{equation}

Our two quantum identities, which reduce to the equality version of the second law and the generalisation of the Jarzynski equation valid for arbitrary initial quasi-classical states, 
can be considered as two independent children of a third, more powerful,
quantum identity 
\begin{equation}
  \label{eq:gibbs-stoch-q}
  {\rm tr}_{W}\! \left[
  \left(\mathcal J_{H'_S +H_W} \compo 
  \Gamma_{SW} \compo \mathcal J^{-1}_{H_S +H_W}
  \right) (\id_S \otimes \rho_{W})
  \right]
  = \id_{S}
\end{equation}
where $\rho_W$ is the initial state of the weight system, $\id_{S}$ is the identity on the system $S$ of interest, $\Gamma_{SW}$  is the completely positive trace preserving map
acting on the joint state
of system and weight which gives its evolution, and we define, as in \cite{aberg2016fully} (but with opposite sign convention),
$
  \mathcal J_{H} (\rho) = 
  e^{\frac \beta 2 H}  
  \rho\,
  e^{\frac \beta 2 H}  
$.

This parent identity can easily be used to derive
a fully quantum, general Jarzynski Equation for arbitrary states
%, simply by applying $\mathcal J^{-1}_{H_S}$ on the left hand side and taking the trace on both sides to obtain 
(Result~\ref{re:qje} in Section \ref{sec:quantumidentities}).
%\begin{equation}
%    {\rm tr}_{SW}\! \left[\left(
%  \mathcal J_{H_W} \compo 
%  \Gamma_{SW} \compo \mathcal J^{-1}_{H_S +H_W}
%  \compo \mathcal J^{-1}_{T\ln \rho_S}
%  \right) 
%  \left( \rho_S \otimes \rho_{W} \right)
%  \right]
%  = Z'_{S}\ .
%\end{equation}
When the input is quasi-classical, it reduces to our generalised Jarzynski Equation for arbitrary initial quasi-classical states.
%, Equation \eqref{eq:genjar}.
Likewise, the parent identity gives a fully quantum version of the second law equality, Result~\ref{re:qsl} in Section \ref{sec:quantumidentities},
%\begin{eqnarray}
%  \nonumber
%  {\rm tr}_{SW}\!\! \left[\left(
%  \mathcal J_
%  {T\ln\rho'_S} \compo 
%  \mathcal J_
%  {H'_S  +H_W} \compo 
%  \Gamma_{SW} \compo \mathcal J^{-1}_{H_S +H_W}
%  \compo \mathcal J^{-1}_{T\ln \rho_S}
%  \right) \! 
%  \left( \rho_S \otimes \rho_{W} \right)
%  \right]
%  \\ \label{eq:quantumsecondlaw}
%  = 1 \hspace{78mm}
%\end{eqnarray}
%where
%$\rho'_S$ is the final state. Equation \eqref{eq:quantumsecondlaw} 
which reduces to the equality version of the second law
%Equation \eqref{eq:secondlaweq} 
when the initial state is quasi-classical.

Now, it is natural to ask what the parent identity, Eq. \eqref{eq:gibbs-stoch-q}, reduces to for quasi-classical states. While Eq. \eqref{eq:gibbs-stoch-q} must necessarily be fulfilled by all thermodynamical processes on quantum states, on quasi-classical states it leads to the following necessary \emph{and sufficient} condition for transition probabilities to be realizable through thermal processes:
\begin{equation}
  \label {eq:gibbs-stoch-qc}
  \sum_{s, w} P(s',w|s)\, 
  \e^{\beta (E_{s'}-E_s +w)}
  =  1\ ,
%\nonumber
\end{equation}
for all $s'$, where  $P(s',w|s)$ is the conditional probability of the final state having energy levels $E_{s'}$, and work $w$ being done by the system, given that the initial state had energy level $E_s$. 
This turns out to be the extension of an important equation from the {\it resource theoretic} approach to quantum 
thermodynamics which finds its origin in ideas 
from quantum information theory \cite{janzing2000thermodynamic,uniqueinfo,dahlsten2011inadequacy,del2011thermodynamic,brandao2011resource,horodecki2013fundamental,aaberg2013truly,brandao2013second,faist2015minimal,skrzypczyk2014work,egloff2015measure, aberg2014catalytic,lostaglio2015description,cwiklinski2015limitations,lostaglio2014quantum, korzekwa2016extraction,halpern2016beyond,woods2015maximum,ng2015limits,halpern2015introducing,wilming2016second, lostaglio2015stochastic, narasimhachar2015low,YungerHalpern15,gallego2015defining,perry2015sufficient,yunger2016microcanonical,goold2016role,vinjanampathy2016quantum}). %\alv{Eq.~\eqref{eq:gibbs-stoch-qc} is necessary as all thermal processes must obey it and sufficient because given a conditional probability that satisfies it, there exists a thermodynamical process that realizes it.}

An overarching idea behind the information theoretic approach is to precisely define what one means by thermodynamics, and thus consider which possible interactions are allowed between a system, a heat bath, and a work storage device, while systematically accounting for all possible resources used in the process. This leads to a definition of thermodynamics known as
Thermal Operations (TO) \cite{Streater_dynamics,janzing2000thermodynamic, horodecki2013fundamental}. This, and its catalytic version~\cite{brandao2013second}, represent the most an experimenter can possibly do when manipulating a system without access to a reservoir of coherence (although one can easily include a coherence reservoir as an ancilla as in \cite{brandao2011resource,aberg2014catalytic, korzekwa2016extraction}). It is thus the appropriate class of operations  
for deriving limitations such as a second law. However, they can be realised by very coarse grained control of the system, and thus also represent achievable thermodynamical operations \cite{perry2015sufficient}.

They also include the allowed
class of operations considered in fluctuation theorems, which include arbitrary unitaries on system and bath. We explain this inclusion in Appendix \ref{app:equiv}. Thermal Operations are thus broad enough to encompass commonly considered definitions of thermodynamics (see \cite{brandao2011resource} for further discussion on this point), including those used in the context of fluctuation relations.

Eq.~\eqref{eq:gibbs-stoch-qc} turns out to completely characterise thermodynamics in the case of fluctuating work.
In information theory, an important class of maps are those which satisfy
the  doubly-stochastic condition, i.e. preservation of the maximally mixed state. In thermodynamics, when there is no work system, any operation must instead preserve the Gibbs state. Equation \ref{eq:gibbs-stoch-qc} is an extension of this
condition to the case where there is a work system which is allowed to fluctuate.
%It is a generalisation of a characterisation of Thermal Operations
%in the case where there is no work storage system:
Taking Eq.~\eqref{eq:gibbs-stoch-qc} with $w=T \log{\frac{Z_S}{Z'_S}}$ gives the Gibbs-preserving condition derived in \cite{janzing2000thermodynamic,Streater_dynamics}. We will show that Eq.~\eqref{eq:gibbs-stoch-qc}
provides a necessary and sufficient condition for thermodynamical transitions between states which are diagonal in the energy eigenbasis, and as a result is a necessary and sufficient condition for work fluctuations.
Using a generalisation of a theorem of Hardy, Littlewood, and Polya \cite{ruch1978mixing}, the condition of Gibbs-preservation was shown in \cite{renes2014work}
to be equivalent to the set of thermodynamical second laws which have recently been proven to be necessary and sufficient condition for quantum thermodynamical state transformations~\cite{horodecki2013fundamental} (c.f. \cite{ruch1978mixing}), 
the so-called thermo-majorization criteria~\cite{ruch1976principle,horodecki2013fundamental}. The latter are conditions on the initial probabilities $P(s)$ and final probabilities $P(s')$ under which one state can be transformed into another. 

Previously, in the resource theory approach, the work storage system had to be taken to be part of the system of interest~\cite{horodecki2013fundamental}.
Using this technique, one can compute the minimal amount of deterministic
work required to make a state transformation~\cite{horodecki2013fundamental} using thermo-majorization. One can also consider
\emph{fluctuating} or \emph{probabilistic} work from the resource theoretic perspective, but thus far, only average work
has been computed \cite{brandao2011resource,skrzypczyk2014work}.
Unresolved, thus far, has been the question of how the information theoretic paradigm fits in with the fluctuation theorem approach. Some further insights have been obtained in attempting to link
the information theoretic approach with the fluctuation theorem approach~\cite{salek2015fluctuations,halpern2015introducing,dahlsten2015equality}, however, how the two paradigms fit together has remained an open question.

Here, we see that one can incorporate fluctuating work explicitly in the resource theoretic approach through Eqs. \eqref{eq:gibbs-stoch-q} and \eqref{eq:gibbs-stoch-qc}. These serve to bring the field of fluctuation theorems fully
into the domain of the information and resource theoretic approach. This is possible because the class of operations which are allowed in the
fluctuation theorem paradigm lies within Thermal Operations.  The latter approach is also able to incorporate not only fluctuations of work, but also of states~\cite{alhambra2015probability,renes2015relative}, 
and we here aim to extend its use to further physically motivated situations.

Finally, it is interesting to compare the power of the relations presented here, with the Jarzynski and Crooks' relations. We do this for one of the simplest
examples, the process of Landauer erasure~\cite{Landauer}, where a bit in an unknown state is erased to the $0$ state. Since the initial state is thermal, one can apply the Jarzynski equality in its standard form.
However, even in this simple case, we find that the new equalities proven here give more information than the
standard Jarzynski and Crooks, in part because one has an independent equality for each possible final energy state.
%equations are consistent with being able to erase at zero work cost. On the other hand, not only is zero-cost erasure ruled out by our second law equality,
One finds a number of additional insights.
Namely, (i) that one needs very large work fluctuations that grow exponentially as the probability that the erasure fails decreases -- the more perfect we demand our erasure to be, the larger the work fluctuations; (ii) it is impossible to even probabilistically extract work in a perfect erasure process; and (iii) that the optimal average work cost for perfect erasure of $T\log{2}$ is only achieved when the work fluctuations associated with successful erasure tend to zero. While these facts are known for protocols that are thermodynamically reversible, we know of no proof that they hold for arbitrary protocols. This simple application is discussed in the Conclusion.
%Namely, (i) that one needs very large work fluctuations that are exponentially large in the probability that the erasure fails -- the more perfect we demand our erasure to be, the larger the work fluctuations; (ii) it is impossible to even probabilistically extract work in an erasure process; and (iii) that the optimal average work cost of $T\log{2}$ is only achieved when\alv{ the weight in the average of the work fluctuations of failing the erasure goes to zero, even if their size is arbitrarily large}. While these are known for the case when protocols are restricted to thermodynamically reversible ones, we know of no proof that they hold for more general protocols. This simple application is discussed in the Conclusion.

The remainder of the paper is structured as follows: in Section \ref{sec:generalized} we define what we consider to be thermodynamics -- namely the set of Thermal Operations in the presence of fluctuating work. This involves three simple conditions on the type of operations which can be performed and we 
find some general constraints they need to obey.
In Section \ref{sec:quantumidentities} we introduce the three fully quantum thermodynamic identities and prove them. In Section \ref{sec:classical}
we show that in the case of states which are diagonal in their energy eigenbasis, these quantum identities each reduce to the
equality version of the second law, a generalisation of the Jarzynski equation, and the extension of the
Gibbs-preservation condition to the case of fluctuating work. In Section \ref{sec:implications} we discuss the implications of our results on determining conditions for state transformations to be possible.
In Section \ref{sec:cq-identities} we show that in the case of the initial state being diagonal in the energy basis and the final state being arbitrary,
the quantum identities reduce to constraints on the expectation values of certain operators with a clear physical interpretation.

\noindent
{\bf Related work:} While this research was nearing completion, \cite{aberg2016fully} appeared on the arXiv. There, a fully quantum Crooks-type identity was derived. This gives a constraint on the quantum state of the 
weight depending on both the
evolution and the time-reversed evolution. As our constraints are on both the system and weight, the two results appear to complement each other without overlap. We relate the two results by proving
a quantum analog of the Crooks relation, of similar form to that in  \cite{aberg2016fully} but applying not just to the weight, but to the system and weight. This is done in Section \ref{sec:crooks}. Some of the other results there can be derived in our framework as well. %\alv{do we want to y something like - we can derive the off-diagonal crooks relations of Aberg? (in response to suggestion 6)}.

\section{Thermal operations with fluctuating work}\label{sec:generalized}

First, let us characterise the type of process/operation that we consider, and show that they are suitably general and implementable to encompass what is commonly considered to be thermodynamics. Our setting consists of a system with Hamiltonian $H_S$, a bath with Hamiltonian $H_B$ initially in the thermal state $\rho_B = \frac 1 {Z_B} \e^{-\beta H_B}$, and an ideal weight with Hamiltonian $H_W=  \int_\mathbb{R} dx\, x \ket{x}\! \bra{x}$, where the orthonormal basis $\{|x\rangle, \forall\, x\in \mathbb R\}$ represents the position of the weight. The operations we consider will allow for the Hamiltonian to change
as we shall see in Subsection \ref{ss:toH}.
Any joint transformation of system, bath and weight is represented by a Completely Positive Trace Preserving (CPTP) map $\Gamma_{SBW}$. 
We only consider maps $\Gamma_{SBW}$ satisfying the following conditions:
\begin{description}
  \item[Unitary on all systems]
  It has an (CPTP) inverse $\Gamma_{SBW}^{-1}$, which implies unitarity: $\Gamma_{SBW} (\rho_{SBW}) = U\rho_{SBW} U^\dagger$.
  
  \item[Energy conservation] 
  The unitary commutes with the total Hamiltonian: $[U,H_S +H_B +H_W] =0$.

  \item[Independence of the ``position" of the weight] The unitary commutes with the translations on the weight: $[U, \Delta_W]=0$.
  
\end{description}
Here $\Delta_W$ is the generator of the translations in the weight and canonically conjugate to the position of the weight $H_W$, that is $[\Delta_W, H_W] = i$. Note that these constraints allow for processes that exploit the coherence of the weight, as in~\cite{skrzypczyk2014work, masanes2014derivation}. We expand on how this in Appendix \ref{app:equiv}, where we show that such coherence can allow us to implement arbitrary unitaries on system and bath.

Both unitarity and energy conservation are fundamental laws of nature, so imposing them is a necessary assumption. Any process which
appears to violate energy conservation in the above sense is in fact energy conserving, one is merely tracing out or ignoring a system which is taking or giving up energy.
%\alv{Redundant: Conversely, an operation which appears to violate energy conservation can be made energy conserving by considering its action on the weight or any other ancilla. }
For example, turning on an interaction Hamiltonian between system, bath and weight can be done via a coherent ancilla as demonstrated in \cite{brandao2011resource}.  One can generate other couplings between the three systems via
the unitary $U$.
Imposing energy conservation on the systems considered, thus ensures that all sources of energy are properly accounted for. Note that while we require that the total process is unitary on the systems, weight and bath,
the operation on system and weight alone will usually not be.

The last condition, independence of the weight position, implies that the reduced map on system and bath $\Gamma_{SB}$ is a mixture of unitaries (Result~1 in~\cite{masanes2014derivation}). Hence the transformation can never decrease the entropy of system and bath, which 
guarantees that the weight is not used as a resource or as an entropy sink. %source of non-equilibrium. 
Independence of the position of the weight can be thought of as a definition of work~\cite{skrzypczyk2014work} and is used in both the information theoretic and fluctuation theorem approaches. In the latter case, the assumption is implicit, since the amount of work is taken to be the difference in energy between the initial and final system/bath. In other words, work
is taken to be a change in energy of either the work system (explicit), or change in energy of the system-bath (implicitly). Conservation of energy ensures that the implicit and explicit paradigms are equivalent. 
Work then is the change in energy of the work system, and does not depend on how much energy is currently stored there, hence the unitary must commute with its translations. In the Appendix, we discuss the connection between different paradigms in more detail, and in particular, show that Thermal Operations is sufficiently general to include the the paradigms typically considered in the context of fluctuation relations.

%Finally, we note that we could allow the use of an ancilla which must be returned in its original state, so called {\it Catalytic Thermal Operations} \cite{brandao2013second}, however, we do not consider this here
%as we will show that in the case of fluctuating work, it is no more powerful than ordinary Thermal Operations with respect to the quantities we are interested in \chris{where do we show this?} \alv{we should get rid of this last paragraph}.\jono{I believe I had an argument but can't reconstruct it now, so yes, okay to remove}

\subsection{Thermal operations with non-constant Hamiltonian}
\label{ss:toH}

%\alv{this is important so maybe needs some expansion?}
Thermal operations are general enough to include the case where the initial Hamiltonian of the system $H_S$ is different than the final one $H'_S$. This is done by including an additional qubit system $X$ which plays the role of a switch (as in~\cite{horodecki2013fundamental}).
Now the total Hamiltonian is
\begin{equation}
  H = H_S\otimes \ket{0}_X%{\rm switch} 
  \bra{0} + H'_S \otimes \ket{1}_X%{\rm switch} 
  \bra{1} +H_B +H_W
\ ,
\end{equation}
and energy conservation reads 
$[V,H] =0$, where $V$ is the global unitary when we include the switch.
We impose that the initial state of the switch is $\ket{0}_X$ and the global unitary $V$ performs the switching
\begin{equation}
  \label{SS}
V \left( 
\rho_{SBW} \otimes \ket{0}_X \bra{0}\right)  V^\dagger = \rho'_{SBW} \otimes \ket{1}_X \bra{1} 
\ ,
\end{equation}
for any $\rho_{SBW}$.
This implies
\begin{equation}
  V=  U
  \otimes\ket{1}_X \bra{0} +
  \tilde U
  \otimes\ket{0}_X \bra{1}
  \ ,
\end{equation}
where $U$ and $\tilde U$ are unitaries on system, bath and weight. The condition $[V,H] =0$ implies
\begin{equation}
  \label{EC}
  U (H_S+H_B+H_W) =
  (H'_S+H_B+H_W) U
  \ .
\end{equation}
Therefore, the reduced map on system, bath and weight can be written as
\begin{equation}
  \label{rmm}
  \Gamma_{SBW} (\rho_{SBW}) =
  U \rho_{SBW} 
  U^\dagger\ ,
\end{equation}
where the unitary $U$ does not necessarily commute with $H_S+H_B+H_W$ nor $H'_S+H_B+H_W$ but satisfies Eq.~\eqref{EC}.

%From now on, and in Result~1, eigenvalues $E_{s}$ and eigenvectors $\ket{s}$ are not necessarily equal to $E_{s'}$ and $\ket{s'}$. Even when there are no degeneracies.

\section{Identities for quantum thermal operations}
\label{sec:quantumidentities}

In this section we derive some fully quantum equalities for Thermal Operations with fluctuating work. In the next section we shall provide the physical meaning of these equalities. 
Thus far, from the information theoretic perspective, some quantum constraints on state transformations are known
i.e. constraints on transformations of the coherences over energy 
levels~\cite{brandao2013second,aberg2014catalytic,lostaglio2015description,cwiklinski2015limitations,lostaglio2014quantum,korzekwa2016extraction}, but none of these constraints apply in the case of fluctuating work.
On the other hand, in the fluctuation theorem approach, no quantum relations are known -- one always assumes that the initial and final states are measured in the energy eigenbasis, thus one is only considering transitions between quasi-classical states.
% This holds even if they are connected by a unitary or other quantum operation \cite{tasaki2000jarzynski,campisi2011colloquium}. 

In what follows we are mostly interested in the joint dynamics of system and weight, which is described by the CPTP map
\begin{equation}
  \label{def Gamma SW}
  \Gamma_{SW} (\rho_{SW}) =
  {\rm tr}_{B} \!
  \left[U
  \left( \rho_{SW} \otimes \frac {e^{-\beta H_B}} {Z_B} 
  \right) U^\dagger \right] .
\end{equation}
It is convenient to define the CP (but not TP) map
\begin{equation} \label{eq:jop}
  \mathcal J_{H} (\rho) = 
  e^{\frac \beta 2 H}  
  \rho\,
  e^{\frac \beta 2 H}  
  \ ,
\end{equation}
whose inverse is
\begin{equation} \label{eq:jop2}
  \mathcal J^{-1}_{H} (\rho) = 
  e^{-\frac \beta 2 H}  
  \rho\,
  e^{-\frac \beta 2 H}  
  \ .
\end{equation}
Using Eqs.~\eqref{EC} and~\eqref{rmm} we obtain
\begin{eqnarray}
  \nonumber &&
  \left(\mathcal J_{H'_S +H_W} 
  \compo \Gamma_{SW} \compo 
  \mathcal J^{-1}_{H_S +H_W}
  \right) (\id_{SW}) 
  \\ \nonumber &=& 
  \mathcal J_{H'_S +H_W}\! \left(
  {\rm tr}_{B}\! \left[ 
  U
  \frac {\e^{-\beta(H_S+H_B+H_W)}} {Z_B} 
  U^\dagger
  \right]\right) 
  \\ \nonumber &=& 
  \mathcal J_{H'_S +H_W}\! \left(
  {\rm tr}_{B}\! \left[ 
  \frac {\e^{-\beta(H'_S+H_B+H_W)}} {Z_B} 
  \right]\right) 
  \\ \label{eee9} &=& 
  \mathcal J_{H'_S +H_W}\! \left(
  \e^{-\beta(H'_S+H_W)}\right) 
  = 
  \id_{SW}\ .
\end{eqnarray}
As mentioned in the previous section, it was proven in~\cite{masanes2014derivation} that the condition $[U, \Delta_W]=0$ implies
${\rm tr}_{W}\! \left[ U\left(\id_{SB} \otimes \rho_W\right) U^\dagger \right] =\id_{SB}$ for any state $\rho_W$.
Proceeding similarly as in Eq.~\eqref{eee9} we obtain
\begin{eqnarray}
  \nonumber &&
  {\rm tr}_W\!
  \left(\mathcal J_{H'_S +H_W} 
  \compo \Gamma_{SW} \compo 
  \mathcal J^{-1}_{H_S +H_W}
  \right) (\id_{S} \otimes \rho_W) 
  \\ \nonumber &=& 
  {\rm tr}_W
  \mathcal J_{H'_S +H_W}\! \left(
  \frac 1 {Z_B}
  {\rm tr}_B\! \left[ 
  U \mathcal J^{-1} _{H_S+H_B+H_W}
  (\id_{SB} \otimes \rho_W)U^\dagger
  \right]\right) 
  \\ \nonumber &=& 
  {\rm tr}_W
  \mathcal J_{H'_S +H_W}\! \left(
  \frac 1 {Z_B}
  {\rm tr}_B\! \left[ 
  \mathcal J^{-1} _{H'_S+H_B+H_W}
  (U\id_{SB} \otimes \rho_W U^\dagger)
  \right]\right) 
  \\ \nonumber &=& 
  {\rm tr}_{BW}\! \left(   
  \frac {\e^{-\beta H_B}} {Z_B} 
  U 
  (\id_{SB} \otimes \rho_W)U^\dagger
  \right) 
  \\ \label{eee10} &=& 
  {\rm tr}_{B}\! \left(   
  \frac {\e^{-\beta H_B}} {Z_B} 
  \, \id_{SB}\right) 
  = 
  \id_{S}\ .
\end{eqnarray}
We thus have:
\begin{result}[Quantum Gibbs-stochastic]
  \label{R1}
If $\Gamma_{SW}$ is a thermal operation, then 
\begin{equation} \label{eq:R1}
  {\rm tr}_{W}\! \left[
  \left(\mathcal J_{H'_S +H_W} \compo 
  \Gamma_{SW} \compo \mathcal J^{-1}_{H_S +H_W}
  \right) (\id_S \otimes \rho_{W})
  \right]
  = \id_{S}
\end{equation}
for any initial state of the weight $\rho_W$.
\end{result}
This is a quantum extension of the Gibbs-preservation condition presented in~\cite{janzing2000thermodynamic,Streater_dynamics}.
The result generalises that in \cite{janzing2000thermodynamic,Streater_dynamics}, not only because it includes work, but also because it is fully quantum.
The details of the quasi-classical generalization to the case of fluctuating work are provided in Section \ref{sec:classical}.

Next, we use the identities $\mathcal J^{-1}_{T\ln \rho} (\rho) =\id$ and ${\rm tr}_S \left[\mathcal J_{T\ln \rho} (\id)\right] = 1$, which hold for any full-rank state $\rho$.
In the case where the initial state $\rho_S$ is not full rank, we can take the limit of a full-rank state.
Now, applying $\mathcal J_{T\ln\rho'_S}$ and taking the trace over $S$ on both sides of Eq.~\eqref{eq:R1} we obtain: 

\begin{result}[Quantum Second Law Equality]
\label{re:qsl}
If $\Gamma_{SW}$ is a thermal operation, then, for every pair of initial states $\rho_S, \rho_W$, we have
\begin{eqnarray}
  \nonumber
  {\rm tr}_{SW}\!\! \left[\left(
  \mathcal J_
  {T\ln\rho'_S} \compo 
  \mathcal J_
  {H'_S  +H_W} \compo 
  \Gamma_{SW} \compo \mathcal J^{-1}_{H_S +H_W}
  \compo \mathcal J^{-1}_{T\ln \rho_S}
  \right) \! 
  \left( \rho_S \otimes \rho_{W} \right)
  \right]
  \\
% \label{e27}
  = 1 \hspace{78mm}
\end{eqnarray}
where 
\begin{equation}
  \label {fs}
  \rho'_S = {\rm tr}_W\!
  \left[ \Gamma_{SW}
  (\rho_S\otimes \rho_W)
  \right]\ ,
\end{equation}
is the final state of the system.
\end{result}
The above result is a quantum generalization of the second law equality, which we will describe in Section \ref{sec:classical}.
Now, applying $\mathcal J^{-1}_{H'_S}$ and taking the trace over $S$ on both sides of Eq.~\eqref{eq:R1} we obtain a quantum generalization of the Jarzynski inequality for general initial states:

\begin{result}[Quantum Jarzynski Equality]
\label{re:qje}
If $\Gamma_{SW}$ is a thermal operation then
\begin{equation}
  {\rm tr}_{SW}\! \left[\left(
  \mathcal J_{H_W} \compo 
  \Gamma_{SW} \compo \mathcal J^{-1}_{H_S +H_W}
  \compo \mathcal J^{-1}_{T\ln \rho_S}
  \right) 
  \left( \rho_S \otimes \rho_{W} \right)
  \right]
  = Z'_{S}
\end{equation}
for every pair of initial states $\rho_S, \rho_W$.
\end{result}

\section{Identities for classical thermal operations}
\label{sec:classical}

We will now go from the fully quantum identities, to ones which are applicable for quasi-classical states (i.e. those considered in fluctuation theorems).
We thus consider the case where there is an eigenbasis $|s\rangle$ for $H_S$ and an eigenbasis $|s'\rangle$ for $H'_S$ such that
\begin{equation}
  \label{cl}
  \Gamma_{SW} 
  (|s\rangle\! \langle s| \otimes |0\rangle\! \langle 0|)
  =
  \sum_{s',w} P(s',w)
  |s'\rangle\! \langle s'|
  \otimes |w\rangle\! \langle w|\ ,
\end{equation}
where $|w\rangle$ are eigenstates of $H_W$.
Note that when $H_S$ or $H'_S$ are degenerate, they could have other eigenbases not satisfying the above.  We say that $\Gamma_{SW}$ is a process which acts on quasi-classical states.
%in the sense that $[\rho_S,H_S]=[\rho_{S'},H_{S'}]=[\rho_W,H_W]=0$.
Also, the ``independence of the position of the weight" allows us to choose its initial state to be $|0\rangle$ without loss of generality.
If we denote by $E_s$ and $E_{s'}$ the eigenvalues corresponding to $|s\rangle$ and $|s'\rangle$, then we can write $\mathcal J_{H_S} (|s\rangle\! \langle s|) = \e^{\beta E_s} |s\rangle\! \langle s|$ and $\mathcal J_{H'_S} (|s'\rangle\! \langle s'|) = \e^{\beta E_{s'}} |s'\rangle\! \langle s'|$.

When Eq.~\eqref{cl} holds, we can represent the thermal operation $\Gamma_{SW}$ by the stochastic matrix
\begin{equation}
  P(s',w|s) = 
  {\rm tr}\! \left[  
  |s'\rangle\! \langle s'| \otimes |w\rangle\! \langle w|\, 
  \Gamma_{SW} 
  (|s\rangle\! \langle s| \otimes |0\rangle\! \langle 0|)
  \right]
  \ .
\label{eq:classicalTO}
\end{equation}
In such a case we have:
\begin{result}[Classical Gibbs-stochastic] $P(s',w|s)$ is a thermal operation mapping quasi-classical states to quasi-classical states if and only if
\label{re:CGS}
\begin{equation}
  \label{eqR4}
  \sum_{s,w} P(s',w|s)\,  
  \e^{\beta (E_{s'}-E_s +w)} = 1
%\nonumber
\end{equation}
for all $s'$.
\end{result}

%The above is a generalization of Gibbs-stochasticity~\cite{janzing2000thermodynamic,Streater_dynamics} to the case where thermodynamical work is included.
\begin{proof}
The proof of the {\it only if} direction follows simply by writing 
Result~\ref{R1} in terms of the matrix of Eq. \eqref{eq:classicalTO}. The {\it if} direction is proven as follows.
 Let us consider a bath with infinite volume in a thermal state at inverse temperature $\beta$. Without loss of generality, the energy origin of the bath can be chosen such that $\langle \mathcal E \rangle_\beta =0$. This and the fact that its heat capacity is infinite (due to the infinite volume) implies that the density of states $\Omega (\mathcal E)$ is proportional to $e^{\beta \mathcal E}$.

Due to energy conservation and invariance of the position of the weight, the joint map of system, bath and weight can be characterised by a map on system and bath $\pi: (s,b) \to (s',b')$ where $(s,b)$ and $(s',b')$ label pairs of system and bath energy levels. 
We construct the map $\pi$ from the given $P(s',w|s)$ in the following way. When the system makes the transition $s\to s'$, a fraction 
$P(s',w= {\cal E-E'} +E_s -E_{s'}|s)$ of the bath states with energy $\mathcal E$ are mapped to bath states with energy $\mathcal E'$, for all values of $\mathcal E$. 
Using the fact that the number of states with energy $\mathcal E$ is $\Omega(\mathcal E) = A\, e^{\beta \mathcal E}$ (for some constant $A$), we will now show that $\pi$ is a permutation.

The number of (final) states in the set $\{(s',b'): \mathcal E_{b'} =\mathcal E'\}$ is $\Omega (\mathcal E')$. And the number of (initial) states $(s,b)$ that are mapped to this set is
\begin{eqnarray*}
  &&
  \sum_{s,\mathcal E} 
  P(s',w= {\cal E-E'} +E_s -E_{s'}|s)\, 
  \Omega(\mathcal E)
  \\ &=&
  \sum_{s,w} 
  P(s',w|s)\, 
  A\,
  e^{\beta (\mathcal E_{s'} -E_s +w +\mathcal E')}
  \\ &=&
  \Omega (\mathcal E')\ ,
\end{eqnarray*}
where in the last line we have used Eq.~\eqref{eqR4}.
Therefore, it is possible to construct a permutation with the mentioned requirements.
\end{proof}

Note that Result~\ref{re:CGS} gives a necessary and sufficient condition that Thermal Operations with a fluctuating weight must satisfy for transformations between quasi-classical states, while the fully quantum Result~\ref{R1} is a necessary condition. 
This last point can be seen by considering an operation that is Gibbs preserving on the system and acts as $\id_W$ on the weight. This clearly satisfies Eq.~\eqref{eq:R1}, yet 
since Gibbs preserving operations are a larger class of operations that Thermal Operations~\cite{faist2015gibbs}, it need not be a Thermal Operation.

  The above is an extension of the Gibbs preservation condition~\cite{janzing2000thermodynamic,Streater_dynamics} to the case where thermodynamical work is included.
%
%Indeed, when the Hamiltonian of the system does not change, setting $w=0$ in Result~\ref{re:CGS} reproduces the aforementioned results and similarly (via the results in \cite{renes2014work}) the thermo-majorization criteria given in \cite{horodecki2013fundamental}. \jono{this is mentioned in two other places}
  %moved part of paragraph to appendix
  When the Hamiltonian of the system does not change, setting $w=0$ in Result~\ref{re:CGS} reproduces the aforementioned result. We discuss the implications of this condition on state transformations in the next section.

In a similar fashion to the previous section, we can write the quasi-classical version of Result~\ref{re:qsl} as
\begin{equation}
  \sum_{s',s,w} P(s',w|s)\,  
  \e^{\beta (f_{s'}-f_s +w)} P(s) 
  = 1\ ,
\end{equation}
where we define the fine-grained free energies
\begin{eqnarray}
  f_s &=& E_s +\frac{1}{\beta}\ln P(s)\ ,
  \\
  f_{s'} &=& E_{s'} +\frac{1}{\beta}\ln P(s')\ .
\end{eqnarray}
In a more compact form:
\begin{result}[Classical Second Law Equality]
\label{R5}
A process on quasi-classical states that acts unitarily on the total system, conserves energy and is independent of the position of the weight satisfies
\begin{equation}
  \label{2le3}
  \left\langle
  \e^{\beta(f_{s'}-f_s +w)}
  \right\rangle =1\ .
\end{equation}
\end{result}
This result follows simply by using Eq. \eqref{eq:classicalTO} in Result~\ref{re:qsl} or directly from Result \ref{re:CGS}.

Due to the convexity of the exponential,  this equality implies the standard second law
\begin{equation}
  \label{2le4}
  \langle f_{s'}-f_s +w \rangle
  \leq 0\ .
\end{equation}
But Eq.~\eqref{2le3} is stronger, since it implies the following infinite list of inequalities
\begin{equation}\label{eq:infbounds}
  \sum_{k=1}^N
  \frac {\beta^k} {k!} 
  \left\langle \left(f_{s'} -f_s + w \right)^k \right\rangle 
  \leq 0\ ,
\end{equation}
where $N$ can be any odd number. 
Note that Eq.~\eqref{2le4} is the $N=1$ case. One can think of Eq. \eqref{eq:infbounds} as providing higher order corrections to the standard second law inequality. 
All the other inequalities have information about the joint fluctuations of $f_s, f_{s'}$ and $w$.
To prove Eq.~\eqref{eq:infbounds} we just note that the residue of the Taylor expansion of the exponential function to any odd order is always negative.

Next we proceed as in Result~\ref{R5}, and obtain the classical version of Result~\ref{re:qje}. Once again, this can be done either by substituting Eq. \eqref{eq:classicalTO} into Result~\ref{re:qje} or proceeding directly from Result~\ref{re:CGS}. %substituting Eq. \eqref{eq:classicalTO} into it.
%\alv{same as in result 5, do we want to refer to eq23 or result 4?}
\begin{result}[Classical Jarzynski Equality]
\label{R6}
A process on quasi-classical states that acts unitarily on the total system, conserves energy and is independent of the position of the weight satisfies
\begin{equation}\label{classJarz}
  \left\langle
  \e^{\beta(w-f_s)}
  \right\rangle 
  = Z'_S\ .
\end{equation}
\end{result}
Note that this version of the Jarzynski equation is valid for any initial state of the system, encoded in the fine-grained free energy $f_s$. For the particular case where the initial state is thermal, we have $\e^{-\beta f_s} = Z_S$ for all $s$, which implies the standard Jarzynski Equality
\begin{equation}
  \left\langle
  \e^{\beta w}
  \right\rangle 
  = \frac {Z'_S} {Z_S}\ .
\end{equation}

\section{Implications for conditions on state transformations} \label{sec:implications}

Result~\ref{re:CGS}, the extension of the Gibbs preserving condition to the case of fluctuating work, is a restriction on what maps are possible in thermodynamics. The Gibbs preserving condition, is likewise a generalisation of the condition for a map to be bistochastic (i.e. preserve the maximally mixed state). This is recovered when the initial and final Hamiltonians of the system are trivial $H_S=H_{S'}=0$, and we set  $w=0$  in Result~\ref{re:CGS}.

For the case of bistochastic maps $\Lambda$, the condition on the map is equivalent to a condition on what state transformations are possible. Namely, for two states $\rho$ and $\rho'$, $\rho'=\Lambda \left(\rho\right)$ if and only if $\rho'$ is majorised by $\rho$  \cite{hardy1952inequalities}. The majorisation condition is as follows. For eigenvalues of $\rho$ and $\rho'$ arranged in non-increasing order and denoted by $\lambda_s$, $\lambda_s'$, we say that  $\rho'$ is majorised by $\rho$ if and only if $\sum_{s=1}^k\lambda_s \geq  \sum_{s=1}^k\lambda_s' \,\, \forall k$. 

For non-trivial Hamiltonians such that $H_S=H_{S'}$ (but still setting $w=0$), Thermal Operations preserve the associated Gibbs state rather than the maximally mixed state. For such sets of allowed operations thermo-majorisation provides a set of conditions that relate initial states to achievable final states (see Figure \ref{fig:w_ss}). These can be considered as a refinement of the second law of thermodynamics, since they constrain the states to which some initial state can evolve under the laws of thermodynamics.

%In the case when \chris{the allowed} maps preserve some state other than the maximally mixed state (in this case, the Gibbs state), thermo-majorisation forms a set of conditions which
%relates initial states to final states (see Figure \ref{fig:w_ss}). These can be considered as a refinement of the second law of thermodynamics, since they constrain the states to which some initial state can evolve under the laws of thermodynamics.

In the case of fluctuating work (no longer requiring that $w=0$), one can now ask whether it is possible to relate the condition on allowed maps given by Result \ref{re:CGS} to a condition akin to thermo-majorisation on the achievable states and work distributions. A simple, and fairly common case is where the values of work that occur are
labeled by $s$ and $s'$ and of the form $w_{ss'}=\alpha_{s'}-\gamma_s$. These can be readily related to the Gibbs-preservation condition of~\cite{janzing2000thermodynamic,Streater_dynamics} (note that processes such as Level Transformations as discussed in the Appendix  are examples of such transformations). For such work distributions, $P\left(s',w_{ss'}|s\right)$ is a thermal operation if and only if $P\left(s'|s\right)\equiv P\left(s',w_{ss'}|s\right)$ satisfy the Gibbs-preservation condition of~\cite{janzing2000thermodynamic,Streater_dynamics} but with the energy levels of the initial and final systems redefined so that the initial energy levels are $E_s+\gamma_s$ and the final are $E_{s'}+\alpha_{s'}$. Determining whether it is possible to convert a state $\rho$ into a state $\sigma$ while extracting work of the form $w_{ss'}=\alpha_{s'}-\gamma_s$ can be done using the thermo-majorization diagrams introduced in \cite{horodecki2013fundamental} as shown in Fig.~\ref{fig:w_ss}. Indeed, when $\alpha_{s'}=-E_{s'}$ and $\gamma_s=-E_s$, the problem reduces to the question of whether $\rho$ majorises $\sigma$.

%In the case of fluctuating work, one can now ask whether one can relate to
%condition on maps to a condition on states and work distributions. A simple, and fairly common case is where the values of work that occur are
%labeled by $s$ and $s'$ and of the form $w_{ss'}=\alpha_{s'}-\gamma_s$. These can be readily related to the Gibbs-preservation condition of~\cite{janzing2000thermodynamic,Streater_dynamics} (note that the processes discussed in Appendix \ref{appB} are examples of such transformations). For such work distributions, $P\left(s',w_{ss'}|s\right)$ is a thermal operation if and only if $P\left(s'|s\right)\equiv P\left(s',w_{ss'}|s\right)$ satisfy the Gibbs-preservation condition of~\cite{janzing2000thermodynamic,Streater_dynamics} but with the energy levels of the initial and final systems redefined so that the initial energy levels are $E_s+\gamma_s$ and the final are $E_{s'}+\alpha_{s'}$. Determining whether it is possible to convert a state $\rho$ into a state $\sigma$ while extracting work of the form $w_{ss'}=\alpha_{s'}-\gamma_s$ can be done using the thermo-majorization diagrams introduced in \cite{horodecki2013fundamental} as shown in Fig.~\ref{fig:w_ss}.

\begin{figure}
\centering
\includegraphics[width=0.95\columnwidth]{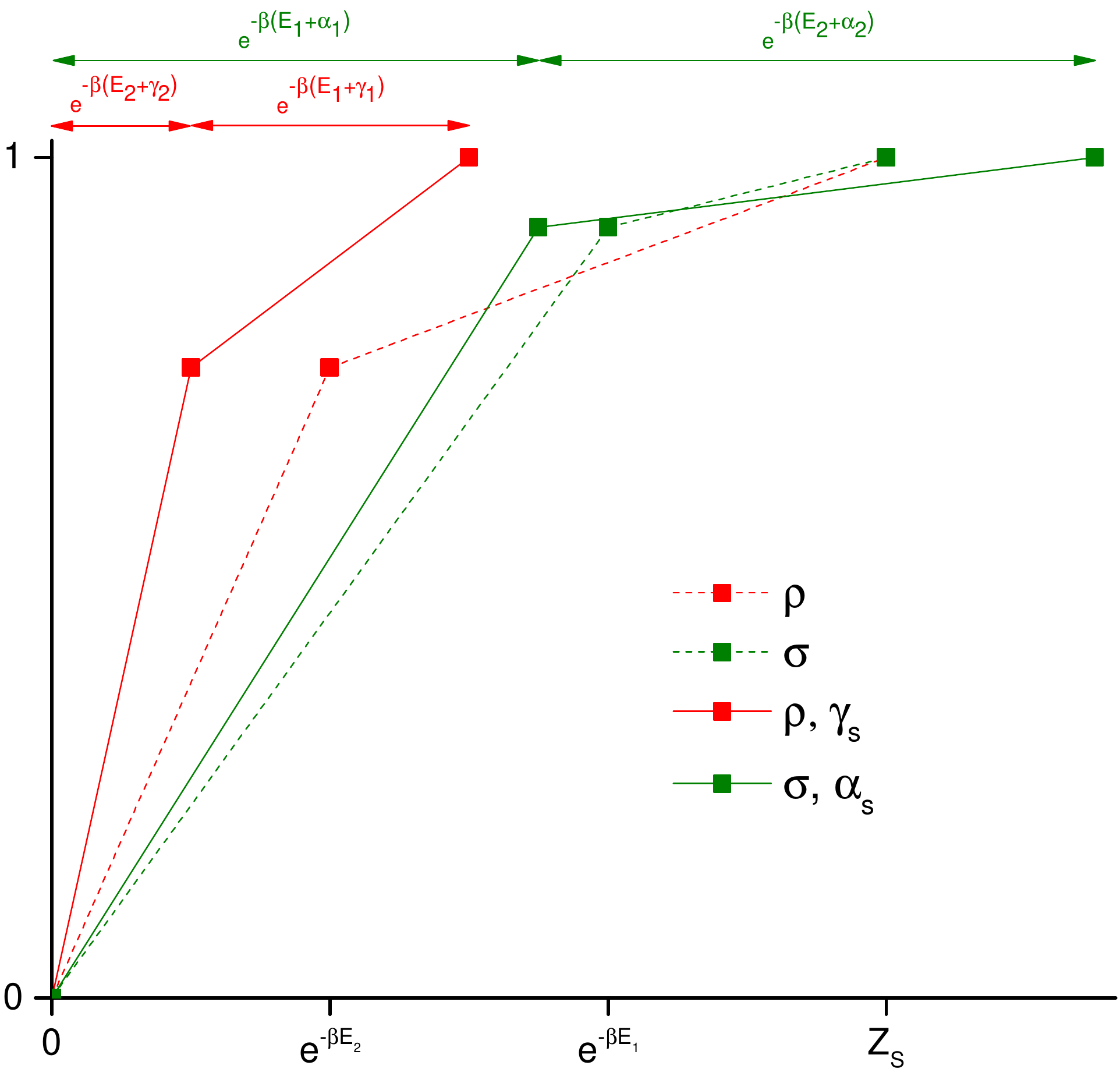}
\caption{Given a system in state $\rho=\sum_{s=1}^{n}p_s\ketbra{s}{s}$ with Hamiltonian $H_S=\sum_{s=1}^{n}E_s\ketbra{s}{s}$, its \emph{thermo-majorization diagram} (see \cite{horodecki2013fundamental} for more details) is formed by first relabeling the pairs of occupation probabilities and energy levels so that $p_1 e^{\beta E_1}\geq p_2 e^{\beta E_2}\geq\dots\geq p_n e^{\beta E_n}$ and then plotting the points $\left\{\sum_{s=1}^{k}e^{-\beta E_s}, \sum_{s=1}^{k}p_s\right\}_{k=1}^{n}$, joining them together to form a concave curve. The figure shows examples for a qubit. In the absence of a work storage system, $\left(\rho,H_S\right)$ can be transformed into $\left(\sigma,H_S\right)$ using a thermal operation if and only if the curve associated with $\rho$ is never below that of $\sigma$. In this example, the curve of $\rho$ crosses that of $\sigma$ so the transformation is not possible. When all values in a work distribution have the form $w_{ss'}=\alpha_{s'}-\gamma_s$, the existence of a thermal operation mapping a quasi-classical state a $\rho$ to quasi-classical state $\sigma$ while producing such a work distribution can be determined by considering the curves associated with $\left(\rho,\sum_{s=1}^{n}\left(E_s+\gamma_s\right)\ketbra{s}{s}\right)$ and $\left(\sigma,\sum_{s=1}^{n}\left(E_s+\alpha_s\right)\ketbra{s}{s}\right)$. In this example, the curve associated with $\rho$ and $\left\{\gamma_s\right\}$ lies above that of $\sigma$ and $\left\{\alpha_s\right\}$ so the transformation from $\rho$ to $\sigma$ is possible with respect to this work distribution. By adjusting $\gamma_s$ and $\alpha_s$ so that both curves are straight lines that overlap, one can make the average work of the transformation equal to the change in free energy and the transformation becomes reversible.} \label{fig:w_ss}
\end{figure}
%\jono{are there other easy cases when the transformation is reversible for the figure (I added the straight line)?} \alv{maybe when $w_{ss'}=E_{s}-E_{s'}$}
%\jono{isn't that  a straight line?}

\section{Classical-quantum identities}
\label{sec:cq-identities}

In classical physics no problem arises from writing joint expectations of observables for the initial and final states of an evolution. For example, this is done in Results \ref{re:CGS}-\ref{R6}.
In general, quantum theory does not allow for this, because a measurement on the initial state will disturb it, and then no longer will it be the initial state.
However, in the case where the measurement is non-disturbing on the initial state, the joint expectation is well-defined, independently of the measurement on the final state.

In what follows we analyze this case, by imposing that both the system and weight are initially quasi-classical.
%\begin{eqnarray}
%  \label{e29}
%  \left[ \rho_S, H_S \right] 
%  = 0\ , 
%  \\ \label{e30} 
%  \left[ \rho_W, H_W \right]
%  = 0\ . 
%\end{eqnarray}
We do not impose any constraint on the final state, but define its dephased version by
\begin{equation}
  \Delta'[\rho'_S] =
  \int dt\, \e^{-iH'_S t}
  \rho'_S\, \e^{iH'_S t}\ .
\end{equation}
This dephasing CPTP map projects $\rho'_S$ onto the subspace of Hermitian matrices that commute with $H'_S$.

If $\rho_S = \sum_s P(s) |s\rangle\! \langle s|$ is the spectral decomposition of the initial state, and $|x\rangle$ is an eigenstate of $H_W$, then
\begin{eqnarray}
  \nonumber &&
  (\mathcal J^{-1}_{H_S +H_W} 
  \compo 
  \mathcal J^{-1}_{T\ln \rho_S}) 
  (|s\rangle\! \langle s| \otimes 
  |x\rangle\! \langle x|)
  \\ &=&
  \e^{-\beta(E_s+T\ln P(s) +x)}
  (|s\rangle\! \langle s| \otimes 
  |x\rangle\! \langle x|)\ .
\end{eqnarray}
The following definitions of the {\it free energy operator} are used below
\begin{eqnarray}
  F_S &=& H_S + T\ln\Delta[\rho_S]\ ,
  \\
  F'_S &=& H'_S + 
  T\ln\Delta' [\rho'_S]\ .
\end{eqnarray}
If in the derivation of Result~\ref{re:qsl} we multiply Eq.~\eqref{eq:R1} by $\mathcal J_{T\ln\Delta\rho'_S}$ instead of $\mathcal J_{T\ln\rho'_S}$, we obtain
\begin{equation}
  {\rm tr}_{SW}\!\! \left[\left(
  \mathcal J_
  {F'_S  +H_W} \compo 
  \Gamma_{SW} \compo \mathcal J^{-1}_{F_S +H_W}
  \right) \! 
  \left( \rho_S \otimes \rho_{W} \right)
  \right]
  = 1
\end{equation}
where we have used that $\Delta[\rho_S] = \rho_S$.

Again, independence from the position of the weight allows us to choose $\rho_W = |0\rangle\! \langle 0|$. 
This enables us to write the above equality as
\begin{equation}
  \sum_{s} \e^{-\beta f_s} P(s)\,
  {\rm tr}_{SW}\!\! \left[
  \e^{\beta (F'_S + H_W)}\,
  \Gamma_{SW} \! 
  \left(|s\rangle\! \langle s| \!\otimes\! 
  |0\rangle\! \langle 0|
  \right)
  \right]
  =1
\end{equation}
or, equivalently:
\begin{result}[Classical-Quantum Second Law Equality]
\label{R7}
Consider a process that acts unitarily on the total system, conserves energy and is independent of the position of the weight . If the initial states of system and weight commute with the corresponding Hamiltonians, then
\begin{equation}
  \left\langle
  \e^{\beta F'_S}
  \e^{\beta W} 
  \e^{-\beta F_S}
  \right\rangle
  =1\ .  
\end{equation}
\end{result}

In the same way we have:
\begin{result}[Classical-Quantum Jarzynski Eq]
\label{R8}
\begin{equation}
  \left\langle
  \e^{\beta (W -F_S)}
  \right\rangle
  = Z'_S\ .  
\end{equation}
\end{result}

\section{A Quantum Crooks relation}
\label{sec:crooks}

Here we use our techniques to prove a fully quantum version of the Crooks relation, which is related to that proven in \cite{aberg2016fully} but on the weight and system.
We also derive a classical version directly from our generalised Gibbs-stochastic condition and without the need to assume micro-reversibility.

In relation to the map defined in Eq.~\eqref{def Gamma SW}, we can also define the associated \emph{backwards} CPTP map associated: 
\begin{equation}
  \label{def Theta SW}
  \Theta_{SW} (\rho_{SW}) =
  {\rm tr}_{B} \!
  \left[U^\dagger
  \left( \rho_{SW} \otimes \frac {e^{-\beta H_B}} {Z_B} 
  \right) U \right] .
\end{equation}
Like any CP map, this can be written in Kraus form
\begin{equation}
  \Theta_{SW} (\rho_{SW}) =
  \sum_k A_k \rho_{SW} A_k^\dagger\ .
\end{equation}
The dual of a map is defined as
\begin{equation}
  \Theta_{SW}^* (\rho_{SW}) =
  \sum_k A_k^\dagger \rho_{SW} A_k\ .
\end{equation}
A bit of algebra shows that
\begin{equation}
%  \label{def Theta SW}
  \Theta_{SW}^* (\rho_{SW}) =
  {\rm tr}_{B} \!
  \left[\frac {e^{-\beta H_B}} {Z_B} U
  \left( \rho_{SW} \otimes \id_B
  \right) U^\dagger \right] ,
\end{equation}
from which it follows:
\begin{result}
The forward and backward maps, respectively $\Gamma_{SW}$ and $\Theta_{SW}$, are related via
\begin{equation}
  \mathcal J_{H'_S+H_W} \Gamma_{SW}
  \mathcal J_{H_S+H_W}^{-1/2}
  =
  \Theta_{SW}^*\ .
\end{equation}
\end{result}
This shows that the dual map is analogous to the transpose map that appears in various results of quantum information theory \cite{petz1986sufficient,barnum2002reversing}.
Note that using the classical version of generalised Gibbs-stochasity,
Result~\ref{re:CGS}, we 
can define the map
\begin{equation}
  P_{\rm back} (s,-w|s') = 
  P(s',w|s)\,
  \e^{\beta (E_{s'} -E_s +w)}
  \ .
\label{eq:back}
\end{equation}
One can check that constraint in Eq.~\eqref{eqR4} applied to $P\left(s',w|s\right)$ is equivalent to the normalization of $P_{\rm back} \left(s,-w|s'\right)$, and the normalization of $P\left(s',w|s\right)$ is equivalent to constraint in Eq.~\eqref{eqR4} applied to $P_{\rm back} (s,-w|s')$.  
This constraint implies that $P_{\rm back} (s,w|s')$ is a thermal operation, hence, there is a global unitary generating this transformation. It can also be seen that one can use the unitary that is the inverse of the one that generates $P(s',w|s)$ (although other unitaries may also generate the same dynamics on system and weight).   
$P_{\rm back}(s,-w|s')$ is thus the microscopic reverse of  $P(s',w|s)$. Indeed, by defining the probability of obtaining work $w$ in going from energy level $s$ to $s'$ when the initial state is thermal by $p_{\rm forward}(w,s',s)=P(s',w|s)e^{-\beta E_s}/Z_S$ for the forward process and $p_{\rm back}(-w,s',s)=P_{\rm back}(s',-w|s)e^{-\beta E_{s'}}/Z'_S$ for the reverse, we obtain a Crooks relation
%Indeed, defining the work obtained when going from one energy level to another, when the initial state is thermal 
%$p_{\rm forward}(w,s',s)=P(s',w|s)e^{-\beta E_s}/Z_S$ and   
%$p_{\rm back}(-w,s',s)=P_{\rm back}(s',-w|s)e^{-\beta E_{s'}}/Z'_S$ we obtain a Crooks relation
\begin{align}
\frac{p_{\rm forward}(w,s,s')}{p_{\rm back}(-w,s,s')}=e^{-\beta w} \frac {Z'_S} {Z_S}
\label{eq:ccrooks}
\end{align}
without needing to assume micro-reversibility, which is the starting assumption of \cite{crooks1999entropy,crooks2000path}. One can take $p_{\rm back}(-w,s,s')$ to the RHS of Eq. \eqref{eq:ccrooks} and then sum over $s$ and/or $s'$ to obtain the more standard Crooks relation
\begin{align}
\frac{p_{\rm forward}(w)}{p_{\rm back}(-w)}=e^{-\beta w} \frac {Z'_S} {Z_S}
\end{align}
but  Eq. \eqref{eq:ccrooks} is clearly stronger. %\alv{is it really that much stronger?}

\section{Conclusion}
In this paper we considered thermodynamical operations between a system, a thermal bath and a weight from which one can extract work in a probabilistic way. From a small set of physically motivated assumptions one can show that these operations obey an identity on arbitrary states from which a number of new, or more general equalities can easily be found. The equalities are both of a fully quantum and of a classical nature. One of these, the second law as an equality, is of a much stronger form
than the standard second law.
For example, the saturation of the second-law inequality
\begin{equation}\label{eq:saturation}
  \left\langle f_{s'} -f_s + w \right\rangle = 0\ ,
\end{equation}
implies
\begin{equation} \label{eq:thermrev}
  w = f_s -f_{s'}\ 
  \mbox{ for all } s,s'.
\end{equation}
This regime is called \emph{thermodynamically reversible}, 
and provides the optimal consumption or extraction of work when we take its average $\langle w \rangle$ as the figure of merit.

Outside of the thermodynamically reversible regime, violations of
\begin{equation}
  f_{s'} -f_s + w 
  \leq 0\ ,
\end{equation}
for individual realizations of the process $(s,s',w)$ can occur. Defining the \emph{excess} random variable $v = f_{s'} -f_s + w $, allows us to write Eq.~\eqref{2le3} as
\begin{equation}
  \label{2le2}
  \left\langle \e^{\beta v} \right\rangle = 1\ .
\end{equation}
Recalling that the exponential function gives more weight to the positive fluctuations as compared with the negative ones, we conclude that, outside of the thermodynamically  reversible regime, the negative fluctuations of $v$ must be larger and/or more frequent than the positive ones. In other words: \emph{the violation of the second law is more rare than its satisfaction}. This asymmetry is also articulated by the infinite list of bounds for the moments of $v$ given in Eq. \eqref{eq:infbounds}.

Note that the Gibbs-stochastic condition of Eq. \eqref{eq:gibbs-stoch-qc} gives more information than the Jarzynski equation or second law equality as the number of constraints it imposes is given by the dimension of the final system. In fact, each condition can 
be thought of as a separate second law equality -- a situation which parallels the fact that one has many second laws for individual systems \cite{horodecki2013fundamental,brandao2013second}. This is related to the fact that
in the case with no weight, Gibbs-stochasity is equivalent to these additional second laws given by thermo-majorisation \cite{renes2014work}.

As a concrete and simple example of these conditions, let us take the case of Landauer erasure \cite{Landauer}.  We consider a qubit with $H_S=0$ that is initially in the maximally mixed state and which we want
to map to the $\ket{0}$ state.  Recalling that a positive work value represents a yield, while a negative work value is a cost, we consider a process such that $-w_0$ is the work cost  when erasing $\ket{0}\rightarrow \ket{0}$, 
and $-w_1$ the work cost if the transition $\ket{1}\rightarrow \ket{0}$ occurs. 
We allow for an imperfect process and imagine
that this erasure process happens with probability $1-\epsilon$, while with probability $\epsilon$ we have an error and either  $\ket{0}\rightarrow \ket{1}$ with work yield $\bar{w}_0$ or  
$\ket{1}\rightarrow \ket{1}$ with work yield $\bar{w}_1$. We call such a process {\it deterministic}, because $w$ is determined by the particular transition.

For this scenario, the the Generalised Gibbs-stochastic condition, Eq.~\eqref{eq:gibbs-stoch-qc}, gives two conditions
\begin{align}\label{eq:conditions2}
e^{\beta w_o}+e^{\beta w_1}&=1/(1-\epsilon)
\\
e^{\beta \bar{w}_o}+e^{\beta \bar{w}_1}&=1/\epsilon
\end{align}
We immediately see that to obtain perfect erasure, $\epsilon\rightarrow 0$, then when the erasure fails there must be work fluctuations which scale like $-T\log \epsilon$. Such a work gain happens rarely, but precludes
perfect erasure, and is related to the third law proven in \cite{masanes2014derivation} and is discussed in detail in \cite{richens2016quantum}.

\begin{figure}
\centering
\includegraphics[width=1\columnwidth]{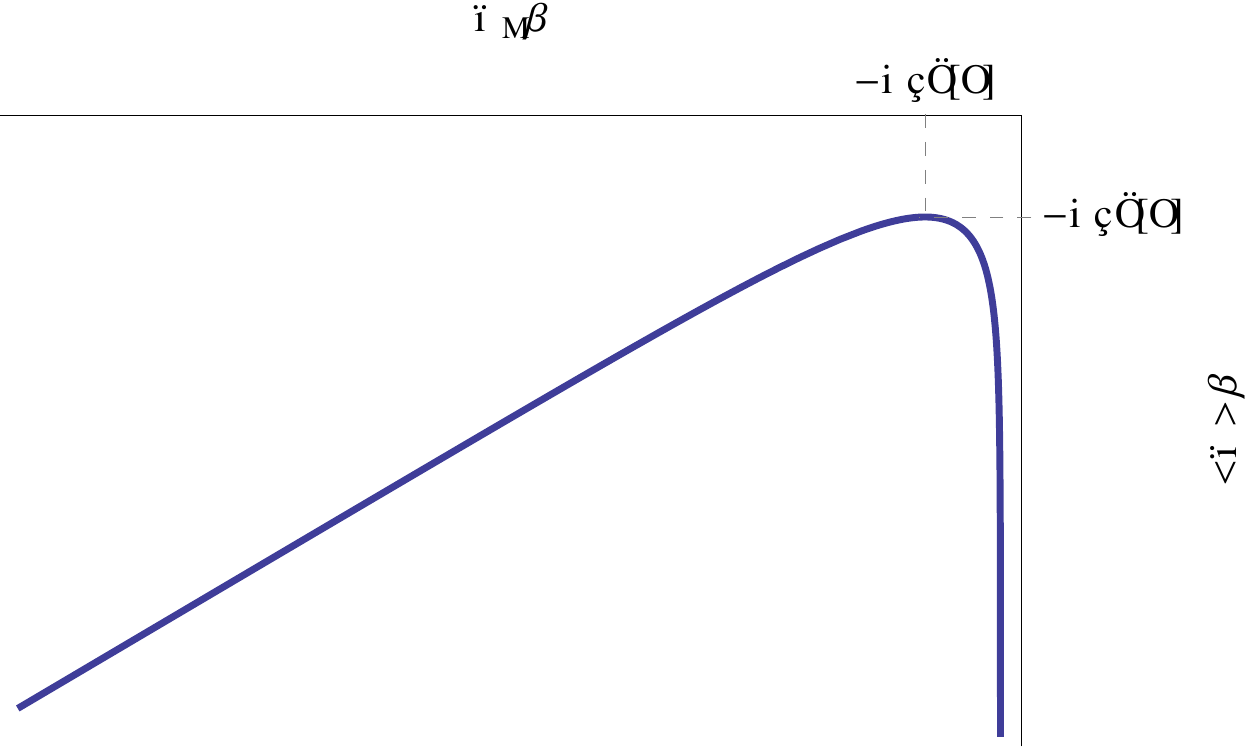}
\caption{ As a simple example of the second law equality, one can think of single qubit erasure. In the limit of perfect erasure, the 2nd law equality reads in this case $e^{\beta w_0}+e^{\beta w_1}=1$, and the average work spent is $\langle w \rangle =\frac{1}{2}(w_0+w_1)$. In the figure, we show the tradeoff between $w_0$ and $\langle w \rangle $ for such perfect erasure. The optimal work value for erasure is the usual Landauer cost at $w_0=w_1=-T \log{2}$. As seen in Eq. \eqref{eq:conditions2}, perfect or near-perfect erasure requires the work cost to fluctuate arbitrarily.
 } \label{fig:1bitCGS}
\end{figure}

In the limit of perfect erasure
we illustrate the work fluctuations in Fig. \ref{fig:1bitCGS}. We easily see that the minimal average work cost of erasure, $T\log 2$, is obtained when the work fluctuations associated with successful erasure are minimal. We also see that no work, not even probabilistically, can be obtained in such a deterministic process.
Since  Eq. \eqref{eq:gibbs-stoch-qc} is not only necessary but also sufficient, we can achieve these work distributions just through the very simple operations described in~\cite{perry2015sufficient}. 
Through this example, one sees that the identities proven here can lead to new insights in thermodynamics, particularly with respect to work fluctuations and their quantum aspects.

 \noindent
{\bf Acknowledgments:} We would like to thank Gavin Crooks for interesting comments on v1 of this paper.  We are grateful to the EPSRC and Royal Society for support. CP acknowledges financial support from the European Research Council (ERC Grant Agreement no 337603).

\bibliographystyle{apsrev4-1}

\clearpage
\widetext
\appendix
\section{Implementing arbitrary unitaries on system and bath} \label{app:equiv}

In this Appendix we show how more traditional derivations of fluctuation theorems, and in particular the Tasaki-Crooks fluctuation theorem \cite{tasaki2000jarzynski,talkner2009fluctuation}, can be obtained within our framework. There, one allows an arbitrary unitary operation on the system and heat bath, and performs an energy measurement before and after this unitary. The difference in energy between the initial and final state is taken to be the amount
of work extracted or expended. Other derivations assume some particular master equation (e.g. Langevin dynamics), which can be thought of as being generated by some particular family of unitaries. On the other hand, we work with Thermal Operations (TO) where we explicitly include a work system (a weight) and
only allow unitaries which conserve the total energy of system, bath and weight. Here, we show that the former paradigms are contained in the one we consider here.

In order to do this, we take the state of the
weight to have coherences over energy levels, which allows us to implement arbitrary unitaries on system and bath. While this is also shown
in \cite{brandao2011resource}, here, we clarify a number of issues in the context of fluctuation relations. We should note that coherences in the weight are only needed if we wish to explicitly model a unitary which creates coherences over energy levels. However, since 
fluctuation theorem results typically require that the initial and final state of the system is measured in the energy eigenbasis, we could consider only unitaries which don't create coherence. Nonetheless, for greater generality, we describe how to implement an arbitrary unitary.

To do this, we show that given the three fundamental constraints we imposed on our allowed operations in Section \ref{sec:generalized}, (unitarity, energy conservation, and independence on the state of the weight) we can give a characterization of the unitary transformation on system and bath. As a consequence, we find that arbitrary unitaries on system and bath can be implemented, and we then show how to obtain the distribution on the work system.

In what follows, it is useful to denote the eigenvectors of the generator of the translations on the weight $\Delta_W = \int\! dt\, t\, \proj t$ by $\ket t$. 
The following result shows that when implementing an arbitrary unitary, the dynamics of the weight is fully constrained, and that the remaining freedom is implicitly characterized by a system-bath unitary:
\begin{lemma}\label{lemma1}
A map $\Gamma_{SBW}$ obeys the three constraints of Section \ref{sec:generalized} (unitarity, energy conservation, and  independence on the state of the weight)  if and only if there is an arbitrary system-bath unitary $V_{SB}$ such that the global unitary on system, bath and weight can be written as
\begin{eqnarray*}
  U_{SBW}
  &=&
  \e^{\i (H'_S +H_B) \otimes \Delta_W}\! 
  \left( V_{SB} \otimes \id_W \right) 
  \e^{-\i (H_S +H_B) \otimes \Delta_W}
  \\ &=&
  \int\! dt \,
  A_{SB} (t) \otimes \proj t
  \ ,  
\end{eqnarray*}
where we define the family of unitaries
\begin{equation*}
  A_{SB} (t) = 
  \e^{\i t (H'_S+H_B)} V_{SB}\,
  \e^{-\i t (H_S+H_B)}\ .
\end{equation*}
\end{lemma}
\begin{proof}

Most of the following arguments do not exploit the system-bath partition. Hence, in order to simplify the expressions, we jointly call them ``composite" $C=SB$, as in $H_C = H_S +H_B$ or $\rho_{CW} =\rho_{SBW}$.
We impose the three fundamental assumptions on the global unitary $U_{CW}$. We start by imposing the {\em independence of the ``position" of the weight.} For this, we note that the only operators which commute with $\Delta_W$ are the functions of itself, $f(\Delta_W)$, and that a complete basis of these functions are the imaginary exponentials $\e^{\i \E\Delta_W}$. Hence, the condition $[U_{CW}, \Delta_W]=0$ implies
\begin{equation}
  \label{first form}
  U_{CW} = \int\! d\E\, A_C(\E ) 
  \otimes \e^{\i \E  \Delta_W}\ ,
\end{equation}
where $A_C (\E )$ with $\E \in \mathbb R$ is a one-parameter family of operators.

Next, we impose {\em energy conservation} 
\begin{equation}
  U_{CW} (H_C +H_W) = (H'_C +H_W) U_{CW}
  \ .
\end{equation}
Note that the equation $[H_W, \Delta_W] = \i$ implies that $[H_W, \e^{\i \E  \Delta_W}] = -\E \, \e^{\i \E  \Delta_W}$ and 
\begin{equation*}
%  \label{first form}
  \int\! d\E 
  \left( A_C(\E ) H_C - H'_C A_C(\E) 
  +\E A_C(\E) \right)
  \otimes \e^{\i \E  \Delta_W}
  =0 \ .
\end{equation*}
This and the linear independence of the operators $\e^{\i \E \Delta_W}$ gives
\begin{equation}
  H'_C A_C(\E )
  =
  A_C(\E) \left(H_C + \E \right)
  \ ,
\end{equation}
for all $\E \in \mathbb R$.
If we translate this equation to Fourier space using
\begin{equation}
  \label{FSA}
  A_C(\E ) = \frac{1}{2 \pi} \int\! dt \, \e^{-\i \E t} A_C(t)
  \ ,
\end{equation}
we obtain
\begin{equation}
  \label{ECF}
  H'_C A_C(t)
  =
  A_C(t) H_C -\i\, \partial_t A_C(t)
  \ .
\end{equation}
The solutions of this differential equation are
\begin{equation}
  \label{solution}
  A_C(t) 
  = 
  \e^{\i t H'_C} V_C\, \e^{-\i t H_C}
  \ ,
\end{equation}
where $V_C$ is arbitrary.

Finally, we impose {\em unitarity} $U_{CW} U_{CW}^\dagger = \id_{CW}$. That is
\begin{equation}
  \id_C \otimes \id_W 
  = 
  \int\! d\E' d\E \, A_C(\E' ) A_C^\dagger (\E)
  \otimes \e^{\i (\E' -\E) \Delta_W}
  \ .
\end{equation}
Using the linear independence of $\e^{\i \E \Delta_W}$ we obtain
\begin{equation}
  \label{dI}
  \int\! d\E \, A_C(\E ) A_C^\dagger (\E +E)
  = \id_{C}\, \delta(E) \ ,
\end{equation}
for all $E\in \mathbb R$. 
If we translate this equation to Fourier space using Eq.~\eqref{FSA} we get $A_C(t) A_C^\dagger (t) = \id_C$, which implies $V_C V_C^\dagger = \id_C$.

Substituting Eq.~\eqref{solution} into Eq.~\eqref{first form} gives
\begin{eqnarray}
  U_{CW} 
  &=& \nonumber 
  \int\! dE\, dt\, \e^{-\i E t} A_C (t)
  \otimes \e^{\i E \Delta_W}
  \\ \label{t form} &=& 
  \int\! dt\, A_C (t)
  \otimes \proj{t}\ .
\end{eqnarray}

An equivalent form can be obtained by using the eigen-projectors of $H_C = \int\! d\E\, \E\, P_\E$ and $H'_C = \int\! d\E' \E' P_{\E'}$. That is
\begin{eqnarray}
  U_{CW} &=& \nonumber 
  \int\! d\E' d\E\, dE\, dt\, 
  \e^{\i (\E' -\E -E) t} 
  \left[ P_{\E'} V_C P_{\E} \right]
  \otimes \e^{\i E \Delta_W}
  \\ &=& \label{third form}
  \int\! d\E' d\E 
  \left[ P_{\E'} V_C P_{\E} \right]
  \otimes \e^{\i (\E'-\E) \Delta_W}
  \\ &=& \label{fourth form}
  \e^{\i H'_C \otimes \Delta_W} A_C\, 
  \e^{-\i H_C \otimes \Delta_W}
  \ .
\end{eqnarray}
If the spectra of $H_C$ and $H'_C$ are discrete, $H_C = \sum_c \E_c \proj{c}$ and $H'_C = \sum_{c'} \E_{c'} \proj{c,}$, then we can write the above as
\begin{equation}
  %\label{third form}
  U_{CW} = \sum_{c',c} 
  \proj{c'} V_C \proj{c}
  \otimes \e^{\i (\E_{c'}-\E_{c}) \Delta_W}
  \ .
\end{equation}

\end{proof}
We stress that there is no constraint on $V_{C}$. 
This type of unitary was used in the context of thermodynamics in ~\cite{brandao2011resource,aberg2014catalytic}.
The above result allows one to obtain an explicit form for the effective map on system-bath (after tracing out the weight)
\begin{eqnarray}
  \nonumber
  \Gamma_{SB} (\rho_{SB}) 
  &=&
  {\rm tr}_W\! \left( U_{SBW} \rho_{SB} \otimes \rho_W U_{SBW}^\dagger \right)
  \\ \label{mix U} &=&
  \int\! dt\, 
  A_{SB} (t) \rho_{SB} A_{SB}^\dagger (t)\,
  \bra t \rho_W \ket t
  \ . \label{eq:unital}
\end{eqnarray}
By noting that $\bra t \rho_W \ket t$ is a probability distribution, we see that the reduced map on system and bath $\Gamma_{SB}$ is a mixture of unitaries (Result~1 in~\cite{masanes2014derivation}). Hence the transformation can never decrease the entropy of system and bath, which 
guarantees that one cannot pump entropy into the weight, which would be a form of cheating. 

Eq.~\eqref{mix U} also implies that, if the initial state of the weight $\rho_W$ is an eigenstate of $\Delta_W$, then the mixture of unitaries only has one term, so that
\begin{result}
If the weight is in a maximally coherent state, that is, an eigenstate of $\Delta_W$ with eigenvalue $t'$, the effect on system and bath is an arbitrary unitary
\begin{equation}
\Gamma_{SB} (\rho_{SB}) =A_{SB}(t') \rho_{SB} A_{SB}(t')^\dagger,
\end{equation}
where
\begin{equation}
A_{SB}(t')=  \e^{\i t' H'_S+H_B} V_{SB}\, \e^{-\i t' H_S+H_B},
\end{equation}
and $V_{SB}$ is as defined in Lemma \ref{lemma1}.
\end{result}
\begin{proof}
In the particular case where of Eq. \eqref{eq:unital} where $\rho_W=\ketbra{t'}{t'}$ is an eigenstate of of $\Delta_W$ so that $\Delta_W \ket{t'}=t' \ket{t'}$, the integral is then 
\begin{align}
 \Gamma_{SB} (\rho_{SB}) &=   \int\! dt\,  A_{SB} (t) \rho_{SB} A_{SB}^\dagger (t)\, \delta(t-t') \nonumber
 \\ &= A_{SB}(t') \rho_{SB} A_{SB}(t')^\dagger.
\end{align}
\end{proof}
 That is, even though this effective map involves tracing out the weight, the result on system-bath is unitary. In addition, this unitary is totally unconstrained, and in particular, it need not be energy-conserving. A typical form for this unitary is $\mathcal T \exp{[\int\! dt\, H_{SB} (t)]}$, where $\mathcal T$ is the time-order operator. 

In summary, thermal operations with fluctuating work can simulate general unitary transformation which do not preserve energy. Hence, statements such as Results \ref{R1}, \ref{re:qsl} and \ref{re:qje} that apply to the first type also apply to the second type. This way we include the operations of the usual derivations of Tasaki-Crooks fluctuation theorems \cite{tasaki2000jarzynski,talkner2009fluctuation}, where arbitrary unitaries on system and bath are allowed.
 
We now outline how one may derive the analogue of fluctuation relations such as Result \ref{R6} in this case, given any unitary dynamics, or mixtures of them. There, the work extracted from system and bath is quantified by measuring their energy before and after the transformation, such that the work takes the form of the random variable $\E'-\E$, where $\E$ is a system+bath energy associated with projector $P_{\E}$. 
The conditional distribution $P(\E' |\E) = {\rm tr}_S\! \left[ P_{\E'} \Gamma_{SB} (P_\E \rho_{SB} P_\E) \right]$ plays the key role. It is known that if the map $\Gamma_{SB}$ is unital (as guaranteed by Eq. \eqref{eq:unital}), then the matrix $P(\E' |\E)$ is doubly-stochastic, and the Jarzynski equality holds, as
\begin{equation} \label{eq:tpm}
\langle e^{-\beta w} \rangle
= \sum_{\E, \E'} e^{\beta (\E-\E')} P(\E' |\E) \frac{e^{-\beta \E}}{Z_{S}Z_B}=\sum_{\E'} \frac{e^{-\beta \E'}}{Z_{S}Z_B}   \sum_{\E} P(\E' |\E)=\sum_{\E'} \frac{e^{-\beta \E'}}{Z_{S}Z_B}=\frac{Z'_S}{Z_S}.
\end{equation}
Other relations, such as the 2nd law equality or Crooks theorem, can be derived analogously too.   Essentially, the double-stochasticity of $P(\E' |\E)$ plays the role of Eq. \eqref{eq:gibbs-stoch-qc} in the derivations of the fluctuation theorems.

Eq. \eqref{eq:tpm} works independently of the state of the weight. In particular, it can be a coherent state $\ket {t'}$, which will make the reduced map on the system unitary.
Hence, the average of the work extracted from system and bath can be equivalently expressed in terms of: (i) measurements on their energy or (ii) shifts in the weight.

An important caveat of the results of this section is that they require the weight to be in a coherent state $\ket{t'}$. While exactly attaining this state is physically impossible, arbitrarily good approximations are possible in principle, allowing for the implementation of maps arbitrarily close to unitary. Here we are showing that when one wants to implement arbitrary unitaries, coherence (understood as a thermodynamical resource) is needed. However, as we have already noted in Section \ref{sec:classical} in the main text, implementing a unitary which merely maps energy eigenstates to energy eigenstates, requires no such coherence. 

%\section{Relation to previous resource-theoretic sets of operations}
%\jono{can we work in the paragraph below?} \label{appB}
Finally, in a different direction, there is a further set of operations that we can include within our framework, and in particular in Result~\ref{re:CGS}. A large part of the literature on resource-theoretic approaches to thermodynamics has been built around a set of operations consisting of sequences of transformations of the energy levels (with an associated work cost), and thermalizations between system and bath. Examples of this are \cite{aaberg2013truly,egloff2015measure,perry2015sufficient}. 
On one hand, the thermalization processes are those for which the work cost vanishes, and consist of a stochastic process (possibly between only two levels) for which we have the following constraint
\begin{equation}
  \sum_{s} P(s'|s)\, 
  \e^{\beta (E_{s'}-E_s)}
  =  1\ ,
%\nonumber
\end{equation}
which is a particular case of Eq.   \ref{eq:gibbs-stoch-qc}
when we take $w=0$ (note that for such thermalization processes the system Hamiltonian remains unchanged).

The level transformation processes, on the other hand, consist of a change of Hamiltonian that leaves the populations of the energy levels invariant $\left(\rho,H_S\right)\rightarrow\left(\rho,H'_S\right)$. Hence, these correspond to stochastic matrices of the form $P\left(s',w|s\right)=\delta_{s, s'}\delta_{E_s-E_{s'},w}$. It can be easily seen that a process like this satisfies Eq.~\eqref{eqR4}. The values of the work distribution that occur in this process are given by the difference between initial and final energy levels, as expected.

\end{document}